\documentclass{article}

\usepackage[utf8]{inputenc}
\usepackage[english]{babel}
\usepackage{csquotes}
\usepackage{amsmath, amssymb, amsfonts, amsthm}

\usepackage{braket}
\usepackage{graphicx}

\usepackage{xcolor}
\definecolor{URLColorBlue}{HTML}{2A1B81}
\definecolor{BlueGreen}{cmyk}{0.85,0,0.33,0}
\definecolor{RawSienna}{HTML}{A0522D}
\definecolor{TealLight}{HTML}{00C798}
\definecolor{Teal}{HTML}{008080}
\definecolor{orange}{HTML}{FF7700}
\definecolor{enji-ro}{HTML}{9D2933}
\definecolor{hana'asagi}{HTML}{1D697C}
\definecolor{seiheki}{HTML}{3A6960}
\definecolor{too'oo}{HTML}{FFB61E}
\definecolor{byakuroku}{HTML}{A5BA93}
\definecolor{konjoo-iro}{HTML}{003171}
\definecolor{hanada}{HTML}{044F67}
\definecolor{koorozen}{HTML}{592B1F}
\definecolor{ootan}{HTML}{FF4E20}
\definecolor{arazome}{HTML}{FFB3A7}
\definecolor{midori}{HTML}{2A603B}
\definecolor{Noshimehana-iro}{HTML}{344D56}

\usepackage{doi}
\usepackage[square,numbers]{natbib}

\usepackage{hyperref}

\hypersetup{colorlinks,
			citecolor=TealLight,
			urlcolor=URLColorBlue,
			linkcolor=hana'asagi,
			pdftitle={Binary linear code weight distribution estimation by random bit stream compression},
			pdfauthor={Alessandro Tomasi, Alessio Meneghetti},
			pdfkeywords={Linear code weight distribution},
}

\theoremstyle{plain}
\newtheorem{proposition}{Proposition}

\newtheorem{algorithm}{Algorithm}

\theoremstyle{definition}
\newtheorem{definition}{Definition}

\theoremstyle{remark}
\newtheorem{remark}{Remark}

\usepackage{authblk}
\author[1]{Alessandro Tomasi\thanks{altomasi@fbk.eu}}
\author[2]{Alessio Meneghetti\thanks{alessio.meneghetti@unitn.it}}
\affil[1]{Security and Trust, Fondazione Bruno Kessler}
\affil[2]{Department of Mathematics, University of Trento}

\begin{document}

\title{Binary linear code weight distribution estimation by random bit stream compression}

\maketitle

\begin{abstract}A statistical estimation algorithm of the weight distribution of a linear code is shown, based on using its generator matrix as a compression function on random bit strings.\end{abstract}

\section{Introduction}

We considered in \cite{Tomasi_Meneghetti_Sala_17} the use of linear code generator matrices as conditioning functions for the output of a random number generator, and showed how the weight distribution of the code determines the distribution of the resulting random variable. Weight distributions are not readily available for many codes; we therefore here consider using random output of a specified quality to estimate the weight distribution, leading to Algorithm \ref{alg:binary_weight_estimation}.

We note from the outset that our solution to the problem of computing the full weight distribution is inefficient, potentially costing more than simply enumerating every weight by brute force. We describe it for the sake of curiosity as an unusual approach; although it can be extended to non-binary codes, the inefficiency of the algorithm prompts us to consider here the binary case only.

Estimates of the weight distribution of linear codes have mostly been deterministic algorithms, for instance based on linear programming \cite{Kasami_Fujiwara_Lin_85, Sala_Tamponi_00}, or approximate analytic bounds \cite{Sole_90, Krasikov_Litsyn_95}. While it is inefficient, Algorithm \ref{alg:binary_weight_estimation} is applicable to any linear code with known generator matrix. In contrast with a deterministic enumeration, it also allows one to make a statistical inference about the whole distribution based on a sample.

There exist probabilistic algorithms for the determination of the minimum weight, generally referred to as information set decoding (ISD). We mostly follow the summary in \cite{Canteaut_Chabaud_98} and \cite{May_Meurer_Thomae_11}, where \cite{Lee_Brickell_88} is identified as the first example, \cite{Leon_88} as an improvement, and \cite{Stern_88} as a somewhat different approach but with the previously best results. These algorithms attempt to solve a rather different problem, but they involve the use of randomness and the syndrome decoding procedure is somewhat reminiscent of our algorithm, so we summarise them and show the common ground in Section \ref{sec:related_work}.

\section{Algorithm derivation}
\label{sec:algorithm_derivation}

We begin by assuming to have access to a random bit generator (RBG) outputting a stream of independent Bernoulli random variables $B$ with fixed and known probability of success. Let the balance of $B$ be defined as
\begin{align*}
	\beta	& = \mathbb{P}(B=0) - \mathbb{P}(B=1)	\\
			& = 1-2\mathbb{E}[B] \,,
\end{align*}
with the so-called bias of $B$ being $\varepsilon/2 = |\beta|/2$.

Let $X$ be a vector of $n$ random bits output by the RBG, and $G$ the generator matrix of a binary $[n,k,d]$ linear code $\mathcal{C}$. We showed in \cite{Tomasi_Meneghetti_Sala_17} how to compute the probability mass function of $Y = GX$, assuming the weight distribution of $\mathcal{C}$ is known. The weight distribution is the sequence
\begin{align}\label{eqn:A_l}
	A_l = \#\set{\mathbf{c} \in \mathcal{C} | w(\mathbf{c})=l}
\end{align}
where $w(\mathbf{c})$ is the Hamming weight of a codeword $\mathbf{c} \in (\mathbb{F}_2)^n$.

The vector $Y \in (\mathbb{F}_2)^k$ is a random variable with probability mass function
\begin{align*}
	\mu_Y		& \in \mathbb{R}^{2^k}, 	\\
	\mu_Y (j)	& = \mathbb{P}(Y = \mathbf{j}) 	\,,
\end{align*}
where in writing $j$ and $\mathbf{j}$ we use the binary representation of integers as vectors $a \in \mathbb{Z}_{2^k}$
\begin{align*}
	\mathbf{a}	& = \left\{a_j \,\middle|\, a = \sum_{j=0}^{k-1} a_j 2^j \right\} \in (\mathbb{F}_2)^k	\,.
\end{align*}

The characteristic function of a random variable $Y$ with mass function $\mu_Y$ is its inverse Fourier transform, which in the binary case is the Hadamard or Walsh transform: 
\begin{align*}
	\chi_Y		& = H \mu_Y
\end{align*}

We can summarise our main reasoning as follows: in Propositions \ref{prop:selection}, \ref{prop:linear_combination}, we show how the weight distribution of $\mathcal{C}$ determines the distribution of $Y$. If the full weight distribution $A_l$ of the corresponding code $\mathcal{C}$ is not known, but its generator matrix $G$ is given, we can apply the compression function to a stream of independent bits of specified balance and then estimate $A_l$ from the characteristic function of the compressed stream $Y$, as described in Algorithm \ref{alg:binary_weight_estimation}.

\begin{proposition}	\label{prop:selection}

	The $b$-th row of $H\mu_Y$ corresponds to the $c$-th row of $H \mu_X$, selected by the code word $\mathbf{c}^T = \mathbf{b}^T G$

\end{proposition}
\begin{proof}

	The $b$-th row of the transform can be written as an expected value in terms of the $b$-th Walsh function
	\begin{align*}
		\chi_Y(b)	& = \mathbb{E}\left[h_b(Y)\right]	\\
					& = \mathbb{E}\left[ (-1)^{\mathbf{b}\cdot Y} \right]	\\
					& = \mathbb{E}\left[ (-1)^{\mathbf{b}\cdot GX} \right]	\\
					& = \mathbb{E}\left[ (-1)^{\mathbf{c}\cdot X} \right]	\\
					& = \mathbb{E}\left[h_c(X)\right]
	\end{align*}

\end{proof}

\begin{proposition}	\label{prop:linear_combination}
	
	If the individual bits of $X$ are i.i.d.\ with balance $\beta$, the $c$-th row of $H \mu_{X}$ equals the characteristic function of a linear combination of $w(c)$ bits, with $w(c)$ the Hamming weight of $c$; hence, the $c$-th row of $H\mu_X$ is equal to $\beta^{w(c)}$, and
	\begin{align}	\label{eqn:chi_and_beta}
		\chi_Y(b) = \beta^{w(c)} \,.
	\end{align}

\end{proposition}
\begin{proof}
	\label{proof:sum}

	On the assumption that each $X(j)$ is i.i.d.\ with balance $\beta$,
	\begin{align*}
		\chi_Y(b)	& = \mathbb{E}\left[ (-1)^{\mathbf{c}\cdot X} \right]	\\
					& = \mathbb{E}\left[ (-1)^{\sum_{i=0}^{w(\mathbf{c})-1} X(i)} \right]
	\end{align*}
	The random variable $S_n = \sum_{i=0}^{n-1} X(i)$ has characteristic $\chi_{S_n}=[1,\beta^n]^T$; hence,
	\begin{align*}
		\chi_Y(b)	& = \mathbb{E}\left[ (-1)^{S_{w(\mathbf{c})}} \right]	\\
					& = (-1)^0 \mu_{S_{w(\mathbf{c})}}(0) + (-1)^1 \mu_{S_{w(\mathbf{c})}}(1) \\
					& = \beta^{w(\mathbf{c})}
	\end{align*}
	
\end{proof}

Note that the zeroth word will always lead to $\chi_X(\mathbf{0})=1$, which we can safely ignore by considering the variable $\chi_Y^* = \chi_Y - \chi_{\mathcal{U}}$, since the uniform distribution over whatever space $Y$ is defined will always have all other entries equal to $0$. Assuming no special ordering of the words can be found, we can treat $\chi_Y^*$ as a random variable to be sampled. The resulting algorithm is as follows:
\begin{algorithm}[Binary linear code weight distribution estimation]	\label{alg:binary_weight_estimation}
	
	Given a binary linear code generator matrix $G$ and a random number generator producing independent output bits with a known, fixed balance $\beta$:
	\begin{enumerate}
		\item	generate $N$ vectors $x=\set{x_i}_{i=1}^{n}$, each of $n$ independent random bits;
		\item	compute $y=Gx$ for all $x$; each $y$ is a sample of a random variable $Y$;
		\item	estimate the mass function of $Y$ using the $N$ resulting samples;
		\item	to recover the exponents estimating the weights $w(c)$, from Eq.\ (\ref{eqn:chi_and_beta}) compute
				\begin{align*}
					\hat{l}(b) = \log_\varepsilon|\chi_Y^*(b)| \,;
				\end{align*}
		\item	round $\hat{l}(b)$ to the nearest integer, $\mathrm{int }(\hat{l}(b))$;
		\item	estimate the weight distribution as
				\begin{align*}
					\hat{A}_l = \#\set{i \in \mathrm{int }(\hat{l}) | i=l}\,.
				\end{align*}
	\end{enumerate}
\end{algorithm}
The above is of course specific to binary codes in the use of $\varepsilon=|\beta|$.

\section{Convergence}
\label{sec:confidence}

Each $\chi_Y(b)$ should converge to $\beta^{w(\mathbf{c})}$ in expectation, from Proof \ref{proof:sum}.  We would like to have an estimate of the speed of this convergence as a function of sample size $s$ and $\beta$. We remark that if we could sample the random variable $\chi_Y(b)$ directly, the variance would decrease with more strongly unbalanced input:
\begin{align*}
	\mathrm{Var}\left((-1)^{\mathbf{b}\cdot Y}\right)	& = \mathbb{E}\left[ \left((-1)^{\mathbf{b}\cdot Y} \right)^2\right] - \left( \beta^{w(\mathbf{c})} \right)^2	\\
				& = 1-\beta^{2w(\mathbf{c})}
\end{align*}
Therefore, taking more strongly unbalanced random input, that is $\beta\approx \pm 1$ but not actually equal to $\pm 1$, should lead to a a finite sample having a tighter spread around the mean value. At the same time, a larger $\beta$ increases the spread between powers of $\beta$, and hence improves the distinguishability of each $\chi_Y(b)$.

In practice, Algorithm \ref{alg:binary_weight_estimation} samples the probability mass function of $Y$ by taking $s$ random samples $\mathbf{x}_s$ from $X$ and computing the sample mean
\begin{align*}
	m_Y(b) = \frac{1}{s}\sum_{j=0}^{s-1} \mathbf{1}_{G\mathbf{x}_j=\mathbf{b}}
\end{align*}
with the indicator function $\mathbf{1}_A=1$ if event $A$ is true, and $0$ otherwise. Note that
\begin{align}	\label{eqn:prob_Gx_equals}
	\mathbb{P}(GX=\mathbf{b}) = \mathbb{P}\left(\bigoplus_{j=0}^{n-1}\gamma_j X(j) = \mathbf{b}\right)
\end{align}
where $\gamma_j$ is the $j$th column of $G$, and we have been somewhat cavalier with notation in using multiplication between $\gamma_j \in (\mathbb{F}_2)^k$ and random $X(j)\in \mathbb{F}_2$ to indicate that either $\gamma_j$ or the zero element of $(\mathbb{F}_2)^k$ will appear in the sum, without defining this more precisely. Our aim is to highlight that there is in general more than one $X$ that will result in the same $\mathbf{b}$, and since each may well have a different Hamming weight, it will have a different probability of occurring as a function of $\beta$. It is therefore not obvious whether it is possible to quantify each $\mu_Y(j)$ individually as a function of $\beta$, which makes it difficult to draw conclusions about the speed of convergence as a function of $\beta$.

For a given $\beta$, we can at least say that Algorithm \ref{alg:binary_weight_estimation} does converge in expectation, as follows. Let $M_a$ be a Bernoulli random variable with probability of success $\mu_Y(a)$:
\begin{align*}
	\mu_{M_a} = [1-\mu_Y(a), \; \mu_Y(a)]^T
\end{align*}
and let $N(s,p)$ be a random variable with binomial distribution over $s$ trials with probability of success $p$. Then the distance between the true value and the sample mean estimator is
\begin{align*}
	\mu_Y(a) - m_Y(a)							& = \frac{1}{s}\sum_{j=0}^{s-1} M_a	\\
												& = \frac{1}{s}N(s,\mu_Y(a))	\\
	\mathbb{E}\left[\mu_Y(a) - m_Y(a)\right]	& = 0 	\\
	\mathrm{Var}(\mu_Y(a) - m_Y(a))				& = \mu_Y(a)(1-\mu_Y(a))
\end{align*}
Each estimator is therefore the difference of two sums of binomial estimators:
\begin{align*}
	\chi_Y(b)-\hat{\chi_Y}(b)	& = \sum_{j=0}^{2^k-1}h_b(j)\left(\mu_Y(j) - m_Y(j)	\right)	\\
								& = \left(\sum_{h_b(j)=1}\mu_Y(j) - m_Y(j)\right) - \left(\sum_{h_b(j)=-1}\mu_Y(j) - m_Y(j)\right)
\end{align*}
Each $\mu_Y(j) - m_Y(j)$ having mean $0$ simplifies this to the sum over $2^k$ binomials with a different variance, which remains unaffected by the subtraction.
\begin{align*}
	\mathrm{Var}(\chi_Y(b)-\hat{\chi_Y}(b))	& = \frac{1}{s^2}\mathrm{Var}\left( \sum_{j=0}^{2^k-1} N(s,\mu_Y(j))	\right)	
\end{align*}
Each $\mu_Y(j)$ being in principle different, we cannot say the resulting sum of binomials will be itself binomial.

\section{Numerical results}
\label{sec:discussion}

Algorithm \ref{alg:binary_weight_estimation} clearly compares unfavourably to the deterministic, brute-force solution: given $G$, one can simply compute every single word of the code by cycling through all $2^k$ possible messages $\mathbf{y}$ and performing a matrix multiplication $\mathbf{y}^T G$ for each of them. This yields the full list of codewords, from which the weight distribution $A_l$ may be immediately deduced. By contrast, as described in Section \ref{sec:confidence} we need to estimate the probability of each $\mathbf{y}$ occurring by repeated sampling, which will take some number of samples $s$ per message, depending on the required accuracy, meaning $s \cdot 2^k$ multiplications $G \mathbf{x}$. We can see an example of this in Figure \ref{fig:BCH_n7_d3_range_pOnes}, showing $\hat{\chi}$ after applying the generator matrix of a BCH(7,4,3) code as a compression function on several random bit streams generated with fixed but different $\beta$, and two separate values of $s$. For reference, the known weight distribution of this code is $A_l = [1, 0, 0, 7, 7, 0, 0, 1]$.

Since the estimators $\hat{l}$ are powers of $\beta$, a small $\beta$ leads to a smaller separation between estimators and hence a poorer overall estimate, for the same sample size (see Figure \ref{fig:BCH_n7_d3_range_pOnes}). Similarly, we expect that for $\beta\approx 1$ the estimators will tend to cluster around the same value, with decreasing variance as seen in Section \ref{sec:confidence}.
\begin{figure}
\centering
	\includegraphics[height=3.8cm]{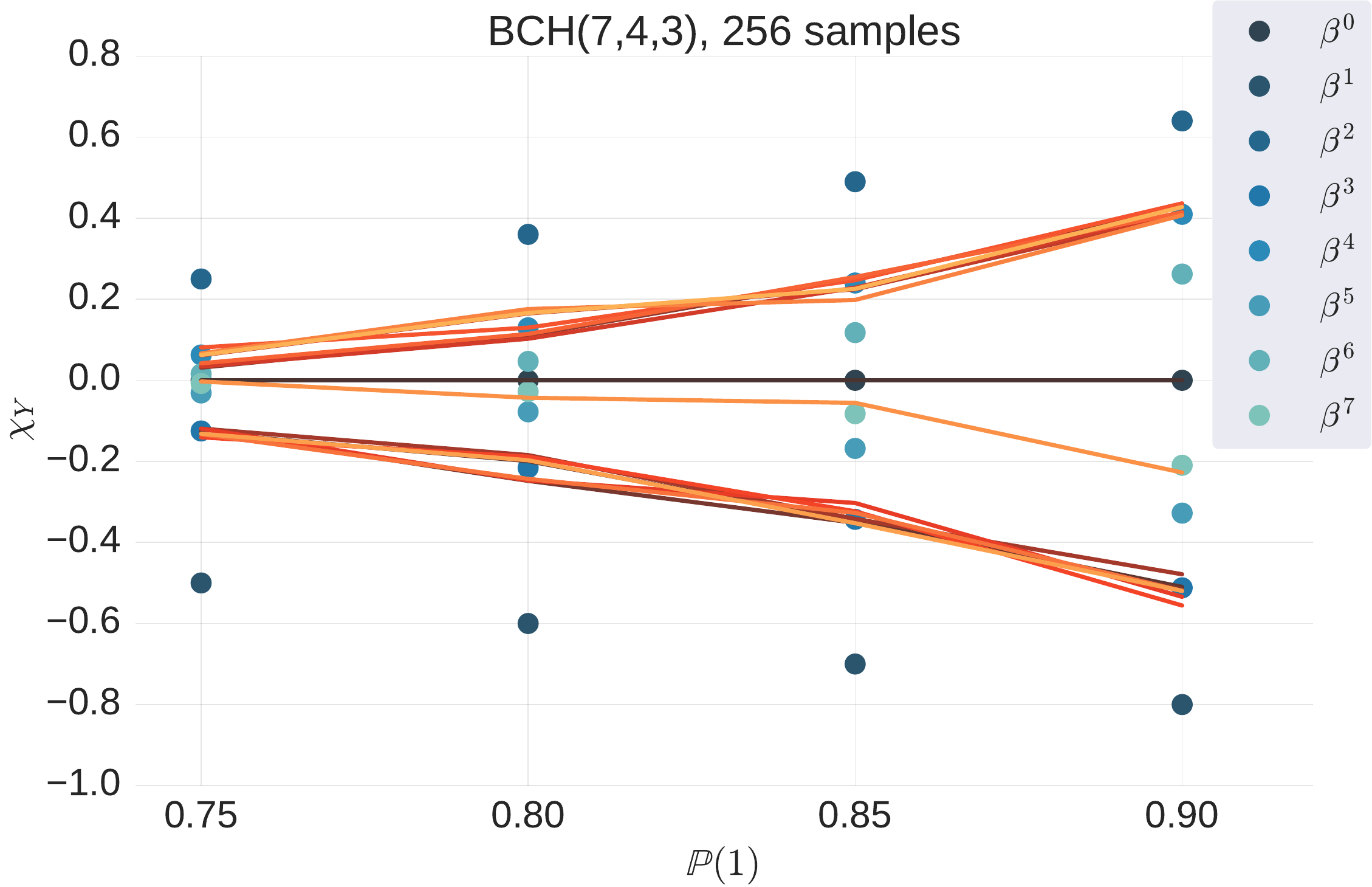}
	\includegraphics[height=3.8cm]{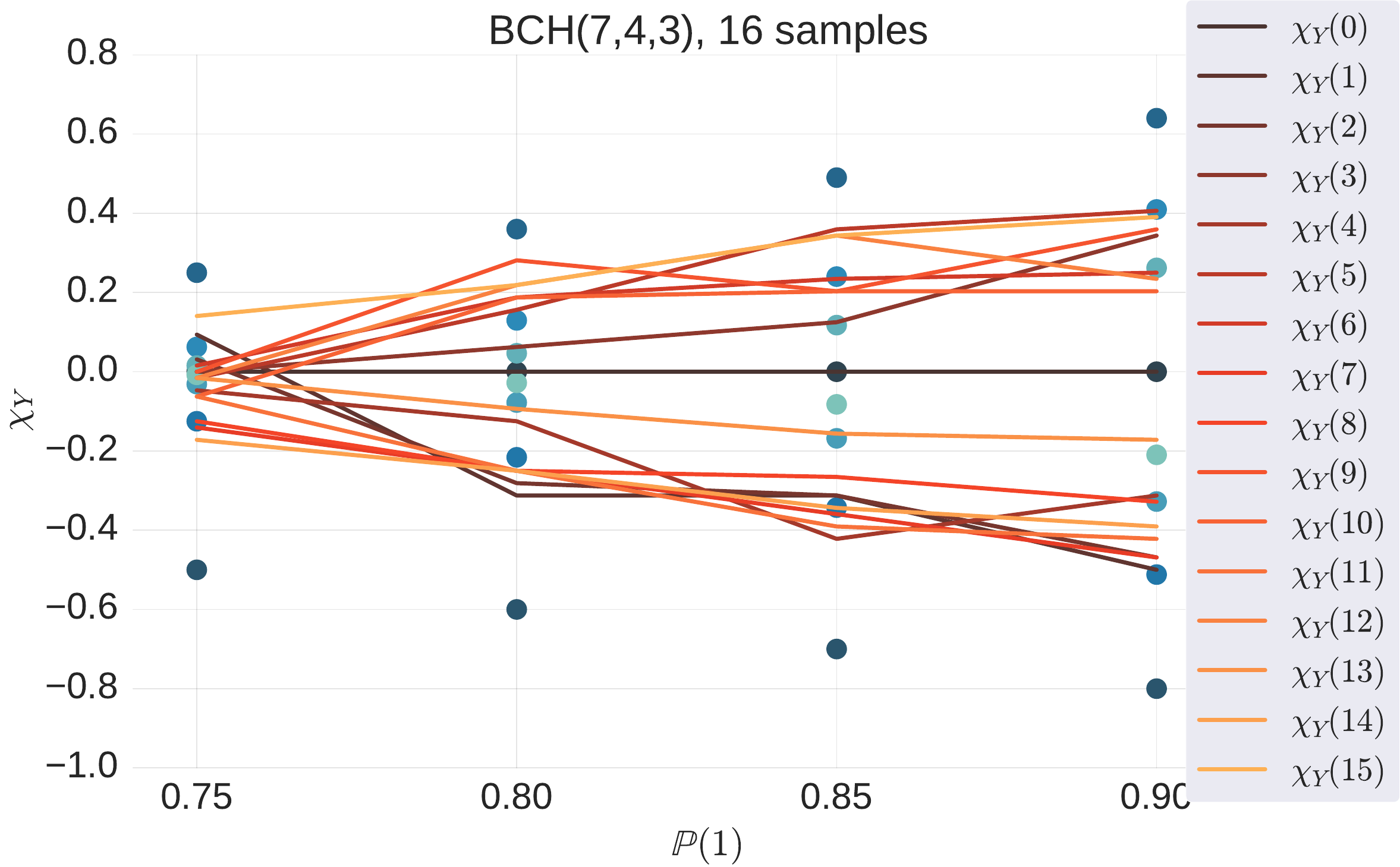}
	\caption{$\chi^*(j)$ as a function of the $\mathbb{P}(1)$ of each of the i.i.d.\ bits of the RBG, compared with all possible powers of $\beta$. These results were obtained using a number of random samples $s=2^{k+4}$ (left) and $s=2^k$ (right), the latter being the required number for a worst-case deterministic algorithm based on computing every codeword exhaustively.}
	\label{fig:BCH_n7_d3_range_pOnes}
\end{figure}
To provide a less unrealistic example we also considered the generator matrix of the BCH(33, 13, 5) code with increasing sample size. Without loss of generality we take a number of samples $s=2^{k-g}$, with $g$ a measure of trade-off between accuracy of the estimate and brute-force equivalence at $g=0$. Results are shown in Figure \ref{fig:BCH_n33_d5_range_samples}; we show a comparison of the estimated weight distribution $\hat{A_l}$ by way of example, and a the distance between the estimated and real weight distribution by comparing the normalized distributions in a total variation sense:
\begin{align*}
	W_l					& = \frac{A_l}{2^k}	\\
	TVD(W_l, \hat{W_l})	& = \frac{1}{2}\sum_{j=0}^{k-1} |W_l(j)-\hat{W_l}(j)|
\end{align*}
\begin{figure}
\centering
	\includegraphics[height=4.2cm]{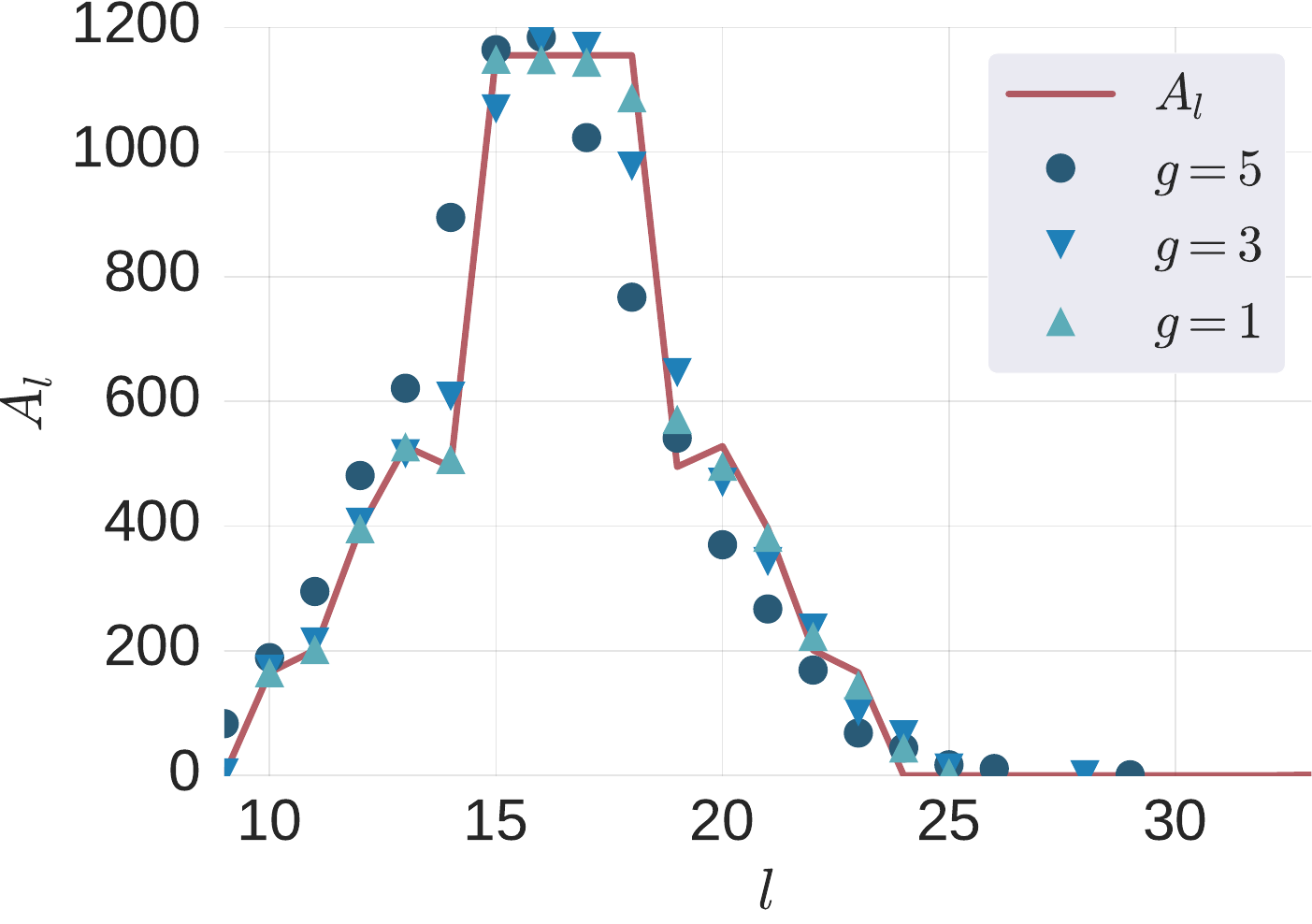}
	\includegraphics[height=4.2cm]{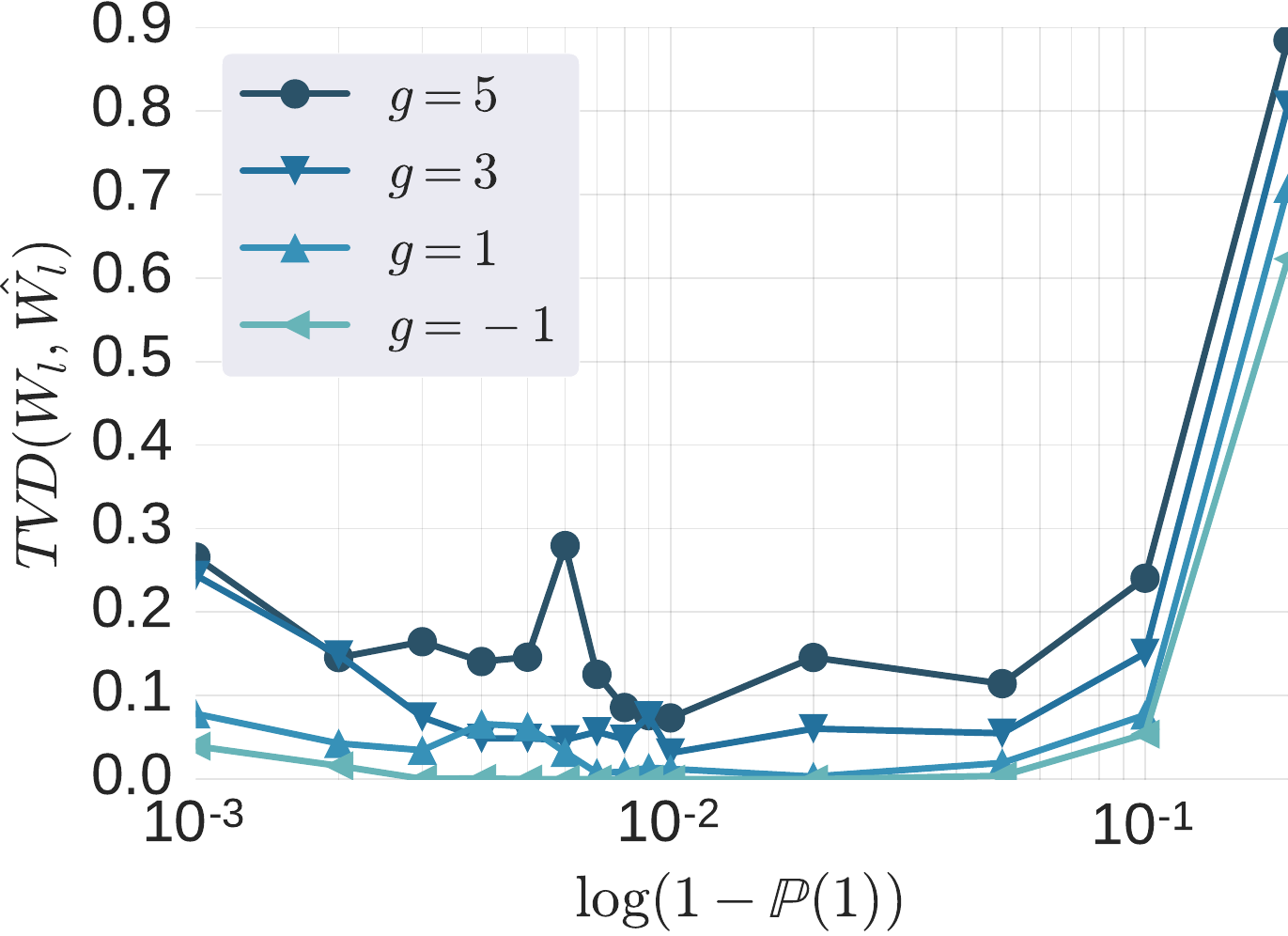}
	\caption{Application of Algorithm \ref{alg:binary_weight_estimation} by conditioning of a bit stream by the generator matrix of the BCH(33, 13, 5) code. Estimates produced with a number of samples $s=2^{k-g}$, where $2^k$ operations would be the brute-force deterministic method cost. Left: visual representation of the estimated weight distribution $\hat{A_l}$ obtained with $\beta=-.95$.}
	\label{fig:BCH_n33_d5_range_samples}
\end{figure}

\section{Related work and improvements}
\label{sec:related_work}

For a linear code $\mathcal{C}$, the minimum weight is equivalent to the minimum distance, $d$. This is of interest to both coding theorists and cryptographers, and the algorithms used are of a somewhat affine nature to the one here proposed, so we briefly discuss the topic - mostly following the summary in \cite{Canteaut_Chabaud_98} and \cite{May_Meurer_Thomae_11}. In one specific instance \cite{Hirotomo_Mohri_Morii_05}, the Stern algorithm for the minimum distance \cite{Stern_88} was extended to the estimation of the whole weight distribution of LDPC codes.

In coding theory, the minimum distance $d$ determines the correction capacity of the code, $t=\lfloor (d-1)/2\rfloor$, and algorithms for finding a specific codeword of minimum weight can be used to correct a received input $x=c+e$, for some $x \in \mathcal{C}$ with errors of small Hamming weight $w(e)\leq t$, by noting that $e$ is the word of minimum weight in the extended code $\mathcal{C}'=\mathcal{C}\oplus x$, and $c$ is the unique closest codeword to $x$.

In cryptography, the security of code-based ciphers has been reduced, in whole or in part, to the problem of finding the minimum weight codeword. For instance, in the McEliece cryptosystem \cite{McEliece_78} the private key is composed of a linear $[n,k,d]$ code generator matrix $G$, an invertible $k \times k$ scrambling matrix $S$, and $n \times n$ permutation matrix $P$; the public key is $\Gamma=SGP$, which generates a code with the same $[n,k,d]$ as $G$; and a ciphertext is computed by encoding a message $m$ and adding an error vector of weight $w(e)\leq t$:
\begin{align*}
	\mu = m\Gamma + e \,.
\end{align*}
The decrypting receiver then computes
\begin{align*}
	\mu P^{-1} = mSG + eP^{-1} \,,
\end{align*}
corrects the permuted errors by applying a decoding algorithm, and hence obtains $m = (mS) S^{-1}$.

Lee and Brickell observed \cite{Lee_Brickell_88} that the best cryptanalytic attack consisted in choosing a set $I$ of $k$ random elements of $\mu$, selecting the corresponding columns of $\Gamma$ as a $k \times k$ matrix $\Gamma_I$, and computing $\tilde{m} = \mu \Gamma_{I}^{-1}$. If $e(I)=0$, this procedure returns the correct $\tilde{m}=m$. They also observe that not only is it costly to invert a random $k \times k$ matrix, there is also a cost associated with checking that the result is correct: if $\tilde{m} \neq m$, then $w(m\Gamma + \tilde{m}\Gamma)\geq d$. One would therefore also have to compute $w(c + \mu(I) \Gamma_{I}^{-1}\Gamma)\leq t$ to be able to claim that $\tilde{m}=m$.

Lee and Brickell further generalized their algorithm to include the possibility that the random set $I$ includes a small number $j$ of errors. Their original algorithm can be written as follows:
\begin{algorithm}[Lee-Brickell \cite{Lee_Brickell_88}]	\label{alg:Lee-Brickell}
	
	Given a McEliece public key $\Gamma$ and a ciphertext $\mu = m\Gamma$:
	\begin{enumerate}
		\item	Select a $k$-element indexing set, $i$ (see Remark \ref{remark:indexing}). Compute
				\begin{align*}
					\Gamma' = \Gamma(\cdot, i)^{-1}\Gamma
				\end{align*}
		\item	Draw $e$ from $E \sim \mathcal{U} \left( \{y \in \mathbb{F}_2^k | w(y)=j\} \right)$, at random without replacement; compute:
				\begin{align*}
					\tilde{m} = \mu + \mu(i)\Gamma' + e \Gamma'
				\end{align*}
				If at any point $w(\tilde{m}) < t$, return $m = \tilde{m}$. Otherwise, restart with a different $i$.
	\end{enumerate}
\end{algorithm}
\begin{remark}[on indexing]\label{remark:indexing}
	There is no particular need to be prescriptive about how a $k$-element indexing is defined. It may be thought of as drawing $k$ times a uniform random variable without replacement from $\mathbb{Z}_n$, though this would require further specification for subsequent draws; or it can be thought of as drawing once without replacement from $I \sim \mathcal{U} \left( \{x \in \mathbb{F}_2^n | w(x)=k\} \right)$, which then requires a slightly more precise meaning to how this is used as an indexing. Since it is not critical to this overview of related work, we skip the details.
\end{remark}

The Lee-Brickell algorithm, as well as others, can be interpreted in the context of information set decoding, an equivalent definition of which can be given in terms of the generator matrix or the parity check matrix - see for instance \cite{Canteaut_Chabaud_98} and \cite{May_Meurer_Thomae_11}.
\begin{definition}[Information Set]

	Let $I$ be a $k$-element indexing of the colmuns of the parity check matrix $H$, and denote by $H_I$ a column-wise permutation of $H$ such that
	\begin{align*}
		H_I	& =(V | W)_I	\\
		W				& = H(\cdot,i)_{i\in I}	\\
		V				& = H(\cdot,j)_{j\notin I}
	\end{align*}
	with $H(\cdot,i)$ the $i$th column of $H$.

	$I$ is called an \emph{information set} for the code $\mathcal{C}$ if and only if $ H_{I} = (Z | \mathbb{I}_{n-k})_{I}$ is a systematic matrix for the code $\mathcal{C}$ - with $\mathbb{I}_k$ the identity matrix of size $k \times k$. The complementary set 
is called a redundancy set.

\end{definition}

From a received message $x=c+e$, the parity check matrix can be used to compute the syndrome $y=He$. The error vector $e$ defines an index set $E$ of columns of $H$ such that $\sum_{j \in E}H(\cdot,j) = e$. The cardinality of this set is unknown, but as long as $\#E\leq t$, decoding will succeed. Syndrome decoding is thus equivalent to finding the set $E$, so that $c=x-\mathbf{1}_E$.

The Lee-Brickell algorithm can be written as follows (see \cite{May_Meurer_Thomae_11}):
\begin{algorithm}[Lee-Brickell - ISD]	\label{alg:Lee-Brickell-ISD}
	
	Let $H$ be a parity check matrix $H$, $U_i \in \mathbb{F}_2^{n \times n}$ be a permutation matrix, $U_\mathcal{G} \in \mathbb{F}_2^{(n-k) \times (n-k)}$ be a Gaussian elimination matrix, and $x$ be a received message. Suppose we look for errors of a specific weight, $w(e)=\eta$.\footnote{This is especially reasonable in the case of a cryptosystem, where one might set $w(e)=t$, as large as possible.}
	\begin{enumerate}
		\item	Draw a random index set $i$ defining a permutation $U_i$. If Gaussian elimination succeeds, we have
				\begin{align*}
					(Z_i | \mathbb{I}) = U_\mathcal{G} (H U_i )
				\end{align*}
		\item	Compute the permuted syndrome $y = U_\mathcal{G}Hx$. Fix a weight $j \leq \eta$. If
				\begin{align*}
					w\left(y \oplus \bigoplus_{b=0}^{k-1}Z_i(\cdot,b)\right) = \eta-j
				\end{align*}
				then we can choose another set $J$ of $\#J=j$ columns from $\mathbb{I}_{n-k}$ to obtain a set $I\cup J$ of columns of $H$ that sum to $y$. Otherwise, restart with a different $i$.
	\end{enumerate}
\end{algorithm}

This algorithm succeeds if the permutation $i$ shifts $e$ exactly so that its first $k$ entries sum to $j$, and the remainder to $\eta-j$.

Further work \cite{Stern_88, Canteaut_Chabaud_98, Finiasz_Sendrier_09, May_Meurer_Thomae_11} makes improvements by more specific assumptions about the non-zero locations of the error vector, such as a contiguous error-free region, or by speed-ups of the Gaussian elimination step.

With respect to the Information Set decoding techniques, 
while Algorithm \ref{alg:binary_weight_estimation} here proposed attempts to solve a different problem and does not require the computation of a systematic form of the matrix $G$, it is interesting to examine the possibility.

Given a systematic form of $G$ written as $(G_s|\mathbb{I}_{k})$, and denoting the $j$th row of $G_s$ as $g_j \in \mathbb{F}_2^{n-k}$, each $Y(j)$ can be written as
\begin{align}	\label{eqn:Yj_systematic}
	Y(j) = B \bigoplus_{i=0}^{n-k-1}g_j(i)X(i)
\end{align}

Each $X(i)$ is a Bernoulli random variable $B$ with balance $\beta$, and compare with Eq. (\ref{eqn:prob_Gx_equals}). Eq. (\ref{eqn:Yj_systematic}) is written to emphasize the fact that the $X(i)$ are the same for all $Y(j)$, which are clearly correlated, plus a single $B$ independent of all the others. It would be possible to modify Algorithm \ref{alg:binary_weight_estimation} to carry out no more than $k(n-k)$ XOR computations for each sample of $Y$, though it would still require drawing $n$ independent random samples of $B$ to compose each $X$.

\section{Conclusion}

We have shown an algorithm that is computationally inefficient but demonstrates an interesting link between the weight distribution of a linear code and the use of its generator matrix as a compression function. We have shown its convergence in theory and in a practical example.

\section*{Acknowledgments}

This research was partly funded by the Autonomous Province of Trento, Call ``Grandi Progetti 2012'', project ``On silicon quantum optics for quantum computing and secure communications - SiQuro''.

The authors would like to thank M.\ Sala and M.\ Piva for insightful discussions on the subject.

\bibliographystyle{abbrvnat}
\bibliography{../RefsA,../RefsACoding,../RefsACrypto}

\end{document}